\documentclass[12pt,letterpaper]{article}

\usepackage{verbatim}
\usepackage{float}
\usepackage{booktabs}
\usepackage{graphicx}
\usepackage{amssymb, amsmath, amsthm}
\usepackage[mathscr]{euscript}
\usepackage{setspace}
\usepackage{natbib}
\usepackage{algorithm}
\usepackage{algorithmic}
\usepackage{xcolor, colortbl,epstopdf}

\newtheorem{lemma}{Lemma}
\newtheorem{proposition}{Proposition}

\definecolor{lightgray}{gray}{0.9}

\bibpunct{(}{)}{;}{a}{,}{,}

\newcommand{\tvec}{\text{vec}}

\newcommand{\bA}{\mathbf{A}}

\newcommand{\bB}{\mathbf{B}}

\newcommand{\bD}{\mathbf{D}}
\newcommand{\bE}{\mathbf{E}}
\newcommand{\be}{\mathbf{e}}

\newcommand{\bh}{\mathbf{h}}
\newcommand{\bI}{\mathbf{I}}
\newcommand{\bK}{\mathbf{K}}

\newcommand{\bP}{\mathbf{P}}
\newcommand{\bQ}{\mathbf{Q}}

\newcommand{\bu}{\mathbf{u}}
\newcommand{\bU}{\mathbf{U}}

\newcommand{\bx}{\mathbf{x}}
\newcommand{\bX}{\mathbf{X}}
\newcommand{\by}{\mathbf{y}}
\newcommand{\bY}{\mathbf{Y}}
\newcommand{\bV}{\mathbf{V}}

\newcommand{\bZ}{\mathbf{Z}}

\newcommand{\mA}{\mathcal{A}}

\newcommand{\mO}{\mathcal{O}}

\newcommand{\mX}{\mathcal{X}}

\newcommand{\mZ}{\mathcal{Z}}

\newcommand{\vect}[1]{\boldsymbol #1}

\newcommand{\vOmega}{\vect{\Omega}}

\newcommand{\vSigma}{\vect{\Sigma}}

\newcommand{\vTheta}{\vect{\Theta}}
\newcommand{\vtheta}{\vect{\theta}}

\newcommand{\e}{\text{e}}
\newcommand{\gvn}{\,|\,}
\renewcommand{\epsilon}{\varepsilon}
\renewcommand{\hat}{\widehat}

\newcommand{\distn}[1]{\mathcal{#1}}

\begin{document}

\title{\mbox{Large Bayesian Tensor VARs with Stochastic Volatility}}

\author{Joshua C. C. Chan \\ Purdue University \and Yaling Qi \\ Purdue University }

\date{September 2024}


\onehalfspace

\maketitle

\thispagestyle{empty}


\begin{abstract}

\noindent We consider Bayesian tensor vector autoregressions (TVARs) in which the VAR coefficients are arranged as a three-dimensional array or tensor, and this coefficient tensor is parameterized using a low-rank CP decomposition. We develop a family of TVARs using a general stochastic volatility specification, which includes a wide variety of commonly-used multivariate stochastic volatility and  COVID-19 outlier-augmented models. In a forecasting exercise involving 40 US quarterly variables, we show that these TVARs outperform the standard Bayesian VAR with the Minnesota prior. The results also suggest that the parsimonious common stochastic volatility model tends to forecast better than the more flexible Cholesky stochastic volatility model.

\bigskip


\bigskip


\end{abstract}

\newpage

\section{Introduction}

Large Bayesian vector autoregressions (BVARs) are now commonly used in macroeconomic forecasting and structural analysis, following early influential papers by \citet*{BGR10}, \citet{CKM09} and \citet{koop13}. The dominant approach to tackle the challenge of parameter proliferation in large systems is to use  shrinkage priors to regularize the variations in VAR coefficients; there is now an extensive literature on various shrinkage priors designed for BVARs.\footnote{The most widely-used shrinkage priors for BVARs are the family of Minnesota priors developed in a series of papers by \citet*{DLS84}, \citet{litterman86} and \citet{KK93, KK97}. Recent additions to this family include \citet*{GLP15} and \citet{chan22}. Another popular family is the adaptive hierarchical shrinkage priors that can be represented as scale mixtures of normals. Examples include the Bayesian Lasso \citep{PC08,KP19}, the normal-gamma prior \citep{GB10,HF19}, the horseshoe prior \citep{CPS10horseshoe,FY19} and the Dirichlet-Laplace prior \citep{DLP15,KH20}. Naturally, one can also combine these adaptive hierarchical priors with the Minnesota priors, as proposed in \citet{chan21}. \citet{HHK24} provide an excellent recent review on the state-of-the-art shrinkage priors developed for macroeconomic forecasting using BVARs.} In contrast, other dimension reduction techniques are relatively unexplored. 

We investigate the usefulness of specifying a low-rank structure on the VAR coefficients for forecasting. More specifically, we follow the approach proposed by \citet{WZLL22} to treat the VAR coefficients collectively as a three-dimensional array or tensor: for a BVAR with $n$ endogenous variables and $p$ lags, we arrange the $n\times n$ coefficient matrices $\bA_1, \ldots, \bA_p$ in the third dimension to construct the third-order tensor $\mathcal{A}\in\mathbb{R}^{n\times n\times p}$. We then model $\mathcal{A}$ using a rank-$R$ CP decomposition. We call these BVARs constructed via the CP decomposition tensor VARs or TVARs.

Using this tensor decomposition, the number of free parameters is reduced from $n^2p$ to $(2n + p)R$. Since the number of free parameters under this tensor decomposition grows linearly in $n$, it is especially suitable for applications with a large number of variables.  This approach is related to the reduced-rank VAR \citep[see, e.g.,][]{CKM11}, which may be viewed as a special case in which the rank of $\mathcal{A}$ is reduced along one of the three possible dimensions.

Departing from the homoskedastic framework in \citet{WZLL22}, we formulate the TVARs using a general stochastic volatility specification, which can represent a wide variety of multivariate stochastic volatility and COVID-19 outlier-augmented models commonly used for BVARs. This more general setup is motivated by the increasing recognition of the importance of allowing time-varying volatility for forecasting macroeconomic and financial variables, especially after the extreme economic turbulence triggered by the COVID-19 pandemic \citep[see, e.g.,][]{LP22,CCMM22}.

We develop efficient estimation procedures for these TVARs with stochastic volatility. In particular, we propose two types of algorithms to sample the components of the CP decomposition of the VAR coefficient tensor $\mathcal{A}$: we can either sample each block of the components jointly as a matrix or sample each column of the component matrix separately. The latter is inspired by the equation-by-equation estimation  approach designed for BVARs proposed in \citet{CCM19} and \citet{CCCM22}, which can drastically reduce the computational burden when the number of endogenous variables $n$ is very large.

We illustrate the methodology using a forecasting exercise that involves 40 US quarterly variables, such as GDP, industrial production, labor market variables and a variety of inflation and interest rates. We consider various TVARs with different stochastic volatility specifications, and compare them to a standard BVAR with the Minnesota prior in an out-of-sample forecasting exercise. The results show that TVARs clearly outperform the standard BVAR, highlighting the usefulness of the low-rank specification for the VAR coefficient tensor. We also find that models with some form of time-varying volatility substantially forecast better than their homoskedastic counterparts. Interesting, the parsimonious common stochastic volatility model of \citet{CCM16} tends to outperform the more flexible stochastic volatility model of \citet{CS05}, suggesting that it is foremost important to capture the strong comovements in the macroeconomic volatilities.

Our paper is closely related to the recent work by \citet{LG23}, who also consider a BVAR with a tensor decomposition. While they focus on the Cholesky stochastic volatility model of \citet{CS05}, we consider a more general setup that can represent a wide range of multivariate stochastic volatility models. Our paper is also related to the emerging literature on modeling multidimensional tensors, where the coefficient tensors are typically parameterized using CP or Tucker decompositions; see, for example, \citet{BCIK23} and \citet{WZL24}.


The rest of this paper is organized as follows. Section \ref{s:VAR} first introduces a general framework of BVARs with a generic time-varying error covariance matrix and discusses how it can be used to represent a variety of popular stochastic volatility and COVID-19 outlier-augmented models. We then outline the proposed approach of arranging the VAR coefficient matrices as a third-order tensor in Section~\ref{s:tensor}. Next, we introduce the efficient sampling algorithms in Section~\ref{s:estimation}. Section~\ref{s:apps} considers a recursive out-of-sample forecasting application that involves 40 US macroeconomic variables. Lastly, Section~\ref{s:conclusions} concludes and briefly discusses some future research directions.

\section{Large Bayesian VARs with Stochastic Volatility} \label{s:VAR}

Let $\by_t$ be an $n\times 1$ vector of endogenous variables at time $t$ for $t=1,\ldots, T.$ Consider the following VAR($p$):
\[
	\by_t = \bA_1 \by_{t-1} + \cdots + \bA_p\by_{t-p} + \bu_t,
\]
 where $\bA_1, \ldots, \bA_p$ are $n\times n$ coefficient matrices. We omit an intercept term for ease of exposition; an intercept or any exogenous variables can be added to the model with minor modifications. Let $\bA = (\bA_1, \ldots, \bA_p)'$ be the  $np\times n$ matrix of VAR coefficients and let $\bx_t = (\by_{t-1}',\ldots, \by_{t-p}')'$ denote a vector of lag variables of dimension $np\times 1$. Then, stacking the observations over $t=1,\ldots, T$, we can rewrite the VAR more succinctly as
\begin{equation}
	\bY = \bX\bA + \bU, \label{eq:var}
\end{equation}
where the matrices $\bY$, $\bX$ and $\bU$ are, respectively, $T\times n$, $T\times np$ and $T\times n$. In a standard homoskedastic VAR, the reduced-form errors $\bu_1,\ldots, \bu_T$ are assumed to be independent and identically distributed (iid) as $\distn{N}(\mathbf{0}_n,\vSigma)$, where $\mathbf{0}_n$ is an $n\times 1$ vector of zeros and $\vSigma$ is an $n\times n$ covariance matrix. However, it is increasingly recognized that some form of time-varying volatility is needed in modeling typical macroeconomic time-series. Early influential papers such as \citet{CS05}, \citet{Primiceri05} and \citet{SZ06} have highlighted the secular variations in volatility. There is now a large empirical literature that demonstrates the importance of allowing time-varying volatility in improving model-fit and forecasting performance in the context of Bayesian VARs; examples include \citet{clark11}, \citet*{KK13}, \citet*{DGG13}, \citet{CR15}, \citet{CP16} and \citet{CE18}.

%

Below we outline a few stochastic volatility models suitable for large BVARs. In particular, the innovations of the VAR in \eqref{eq:var} are now distributed as
\begin{equation} \label{eq:Sigmat}
	\bu_t \sim \distn{N}(\mathbf{0}_n,\vSigma_t),
\end{equation}
where $\vSigma_t$ is a generic time-varying covariance matrix. 

One of the first stochastic volatility models designed for large BVARs is the common stochastic volatility proposed in \citet*{CCM16}. Their model is motivated by the empirical observation that the estimated time-varying error variances of many macroeconomic variables have broadly similar low-frequency movements. A parsimonious way to model these comovements is to introduce a time-varying latent factor to scale the error covariance matrix via
\begin{equation}\label{eq:csv}
	\vSigma_t = \e^{h_t}\vOmega, 
\end{equation}
where $\vOmega$ is a time-invariant covariance matrix. The log-volatility $h_t$ is modeled using a zero-mean stationary AR(1) process:
\begin{equation}\label{eq:h}
	h_t = \phi h_{t-1} + u_t^h, \quad u_t^h\sim\distn{N}(0,\sigma_h^2),
\end{equation}
for $t=2,\ldots, T$, where $|\phi|<1$ and the process is initialized as $h_{1}\sim\distn{N}(0,\sigma_h^2/(1-\phi^2))$. The common stochastic volatility model specified in \eqref{eq:csv}--\eqref{eq:h} can be extended to incorporate other useful features. For example, \citet{chan20} introduces a general framework that can accommodate heavy-tailed, heteroskedastic and serially dependent innovations. 

Another widely-used stochastic volatility specification for BVARs is the Cholesky stochastic volatility model---based on the modified Cholesky decomposition of the covariance matrix---developed in \citet{CS05}. More specifically, consider the decomposition of $\vSigma_t $ via
\begin{equation} \label{eq:VAR-SV}
	\vSigma_t = \bB_0^{-1} \bD_t (\bB_0^{-1})',
\end{equation}
where $\bB_{0}$ is an $n \times n$ lower triangular matrix with ones on the diagonal and $\bD_t = \text{diag}(\e^{h_{1,t}}, \ldots, \e^{h_{n,t}})$. Each element of the vector $\bh_t = (h_{1,t}, \ldots, h_{n,t})'$ is modeled via an independent autoregressive process:
\begin{equation}\label{eq:hit}
	h_{i,t} = \mu_i + \phi_i(h_{i,t-1}-\mu_i) + u_{i,t}^h, \quad u_{i,t}^h \sim \distn{N}(0, \sigma_{i}^2)
\end{equation}
for $t=2,\ldots, T$, where $h_{i,1}$ is initialized as $h_{i,1}\sim\distn{N}(\mu_i,\sigma_i^2/(1-\phi_i^2))$ for $i=1,\ldots, n.$

This Cholesky stochastic volatility specification is more flexible than the common stochastic volatility model in \eqref{eq:csv}, since the former contains $n$ stochastic volatility processes and can accommodate more complex covolatility patterns. This flexibility, however, comes at a cost of higher model complexity. Whether this stochastic volatility specification forecasts better than alternatives in large systems is an empirical question. 

Another difference between the Cholesky stochastic volatility and the common stochastic volatility is that the latter is order-invariant---i.e., parameter estimates are invariant to reordering the endogenous variables in $\by_t$---whereas the former is not. One practical implication is that forecasts from the Cholesky stochastic volatility model could differ substantially across different variable orderings, as documented in \citet{ARRS23} using a similar model of \citet{Primiceri05}. 

The reason why the Cholesky stochastic volatility is not order-invariant is partly due to the use of the lower triangular parameterization of $\bB_0$ in \eqref{eq:VAR-SV}. Motivated by this simple observation, \citet{CKY24} extend the model by relaxing this lower triangular assumption and instead specify $\bB_0$ to be any non-degenerate square matrix. They prove that the model is order invariant. Moreover, based on the results in \citet{BB22}, $\bB_0$ is also identified up to permutations and sign switches.

The extreme movements in many macroeconomic variables at the onset of the COVID-19 pandemic have motivated much recent work on modeling outliers in macroeconomic time-series. An example is the outlier-augmented stochastic volatility model proposed by \citet{CCMM22}, which builds on the Cholesky stochastic volatility model and the discrete mixture representation for the innovations introduced in \citet{SW16}. This outlier-augmented model can also be represented using~\eqref{eq:Sigmat}.


There are many other multivariate stochastic volatility models for $\vSigma_t$. \citet{CM23} and \citet{chan24} provide two recent reviews on BVARs with a wide range of stochastic volatility and outlier-augmented specifications.

\section{Tensor Decomposition of VAR Coefficients} \label{s:tensor}

For high-dimensional settings, there are two key challenges related to the proliferation of the VAR coefficients. First, the number of VAR coefficients increases quadratically in $n$. In many large-scale applications, there are far more VAR coefficients than the number of observations, which makes it necessary to regularize these VAR coefficients. Second, due to the proliferation of  VAR coefficients, sampling them tends to be very computational intensive. 

These two related challenges are typically tackled separately in the literature. For instance, Bayesian shrinkage priors are widely used to regularize the VAR coefficients. These include the family of Minnesota priors \citep*{DLS84, litterman86, KK93, KK97, GLP15,chan22} and various adaptive hierarchical shrinkage priors \citep{HF19,KP19,KH20,chan21}. The computational challenge is addressed by developing efficient MCMC or variational methods to sample the large number of VAR coefficients \citep{CCM19,CCCM22,GKP23,BBB24}. We instead take an alternative approach that tackles these two challenges simultaneously by imposing a low-rank structure on the VAR coefficients. 

One possibility is the reduced-rank VAR \citep[see, e.g.,][]{CKM11}, in which the matrix $\bA' = (\bA_1, \ldots, \bA_p)$ is assumed to have a reduced rank $R<n$. That is, the dimension of the column space of the VAR coefficient matrices $\bA_1, \ldots, \bA_p$ has rank $R$. Alternatively, one could impose a low-rank structure on $(\bA_1', \ldots, \bA_p')$ or $(\text{vec}(\bA_1), \ldots, \text{vec}(\bA_p))$, whose ranks correspond to the dimensions of the row space and vectorized matrix space of the VAR coefficient matrices, respectively. Each of these options essentially reduces the dimension along one of the three different directions.

This motivates us to follow the approach in \citet{WZLL22} to treat the VAR coefficients collectively as a three-dimensional array or tensor.  That is, we arrange the $n\times n$ coefficient matrices $\bA_1, \ldots, \bA_p$ in the third dimension to form the third-order tensor $\mathcal{A}\in\mathbb{R}^{n\times n\times p}$. With this tensor representation, we then use a rank-$R$ CP decomposition to construct $\mathcal{A}$:
\begin{equation}\label{eq:CP}
	\mathcal{A} = \sum_{r=1}^R \vtheta_1^{(r)}\circ  \vtheta_2^{(r)} \circ \vtheta_3^{(r)},
\end{equation}	
where $\vtheta_1^{(r)}, \vtheta_2^{(r)}\in\mathbb{R}^{n}$ and $\vtheta_3^{(r)} \in\mathbb{R}^{p}$. Each component $\vtheta_1^{(r)}\circ  \vtheta_2^{(r)} \circ \vtheta_3^{(r)}$ is a rank-1 third-order tensor whose $(i,j,k)$ element is the product of the $i$-th, $j$-th and $k$-th elements of $\vtheta_1^{(r)}, \vtheta_2^{(r)}$ and $\vtheta_3^{(r)}$, respectively. Using this tensor decomposition with a small $R$, the number of parameters is reduced from $n^2p$ to $(2n+p)R$. We refer the readers to \cite{TB09} for a general introduction to tensors and their operations.

\section{Bayesian Estimation} \label{s:estimation}

In this section we describe the priors and outline the posterior simulator. In particular, we focus on the sampling of the components $\vtheta_1^{(r)}, \vtheta_2^{(r)}$ and $\vtheta_3^{(r)}, r=1,\ldots, R$. To that end, let $\vTheta_1 = (\vtheta_1^{(1)}, \ldots, \vtheta_1^{(R)})$ and similarly define $\vTheta_2$ and  $\vTheta_3$. The dimensions of $\vTheta_1, \vTheta_2$ and $\vTheta_3$ are, respectively, $n\times R$, $n\times R$ and $p\times R$. For later reference, stack $\vtheta_1 = \text{vec}(\vTheta_1)$, $\vtheta_2 = \text{vec}(\vTheta_2')$ and $\vtheta_3 = \text{vec}(\vTheta_3')$. For reasons that will become transparent later, note that $\vtheta_1$ is constructed by stacking the columns of $\vTheta_1$, whereas $\vtheta_2$ and $\vtheta_3$ are formed by stacking the rows of $\vTheta_2$ and $\vTheta_3$, respectively.

Next, consider the following independent Gaussian priors on $\vtheta_1, \vtheta_2$ and $\vtheta_3$:
\begin{equation} \label{eq:prior}
	\vtheta_j \sim\distn{N}(\vtheta_{j,0}, \bV_{\vtheta_j}), \quad j=1,\ldots, 3.
\end{equation}
Here we consider simple Gaussian priors, but any adaptive hierarchical shrinkage priors that have a  conditionally Gaussian representation---such as the normal-gamma prior, the horseshoe prior or the Dirichlet-Laplace prior---can be used.

In what follows, we derive the full conditional posterior distributions of $(\vtheta_1 \gvn \bY, \vtheta_{2}, \vtheta_{3}, \vSigma)$, $(\vtheta_2  \gvn \bY, \vtheta_{1}, \vtheta_3, \vSigma)$ and $(\vtheta_3 \gvn \bY, \vtheta_{1}, \vtheta_2, \vSigma)$, where $\mathbf{\vSigma}=\text{diag}\left(\vSigma_1, \ldots, \vSigma_T\right)$.  

\subsection{Sampling of $\vtheta_1$}

For sampling $\vtheta_1$, first let
\begin{align*}
	\vTheta_{-1} & = (\vtheta_3^{(1)}\otimes\vtheta_2^{(1)}, \ldots, \vtheta_3^{(R)}\otimes\vtheta_2^{(R)}), \\
	\vTheta_{-2} & = (\vtheta_3^{(1)}\otimes\vtheta_1^{(1)}, \ldots, \vtheta_3^{(R)}\otimes\vtheta_1^{(R)}), \\
	\vTheta_{-3} & = (\vtheta_2^{(1)}\otimes\vtheta_1^{(1)}, \ldots, \vtheta_2^{(R)}\otimes\vtheta_1^{(R)}).
\end{align*}
Then, the mode-1 matricization of $\mathcal{A}$ can be written as \citep[see, e.g.,][]{TB09}:
\begin{equation} \label{eq:A1}
	\mathcal{A}_{(1)} = \bA' = \vTheta_1 \vTheta_{-1}'.
\end{equation}
Furthermore, let  $\mX\in \mathbb{R}^{T\times n \times p}$ denote the third-order tensor constructed by stacking the $n\times p$ matrices $(\by_{t-1},\ldots, \by_{t-p}), t=1,\ldots, T,$ along the first dimension so that $\tvec(\mX_{t,:,:}) = \bx_t$. It is easy to verify that its mode-1 matricization is the $T\times np$ matrix $\bX$, i.e., $ \mX_{(1)}=\bX$. Hence, combining \eqref{eq:var} and \eqref{eq:A1}, we have
\begin{equation} \label{eq:var_A1}
    \bY = \mX_{(1)}\vTheta_{-1}\vTheta_1'+\bU.
\end{equation} 
It follows that 
\[
   \tvec(\bY')=(\mX_{(1)}\vTheta_{-1} \otimes \bI_{n})\tvec(\vTheta_1)+\tvec(\bU'),
\]
where $\tvec(\bU') \sim \mathcal{N}(\mathbf{0}_{Tn}, \vSigma)$ with $\mathbf{\vSigma}=\text{diag}\left(\vSigma_1, \ldots, \vSigma_T\right)$. Given the Gaussian prior on $\vtheta_1 = \tvec(\vTheta_1)$ specified in \eqref{eq:prior}, the full conditional posterior of $\vtheta_1$ is given by
\[
	(\vtheta_1 \gvn \bY, \vtheta_{2}, \vtheta_{3}, \vSigma) 
	\sim \mathcal{N}\left(\hat{\vtheta}_1, \mathbf{K}_{\vtheta_1}^{-1}\right), 
\]
where
\begin{align*}
	\bK_{\vtheta_1} & = \bV_{\vtheta_1}^{-1} + (\vTheta_{-1}'\mX_{(1)}'\otimes \bI_{n})\vSigma^{-1}(\mX_{(1)}\vTheta_{-1} \otimes \bI_{n})\\
\hat{\vtheta}_1   & = \bK_{\vtheta_1}^{-1}\left(\mathbf{V}_{\vtheta_1}^{-1} \vtheta_{1,0}
		+(\vTheta_{-1}'\mX_{(1)}'\otimes \mathbf{I}_n)\vSigma^{-1}\tvec(\bY')\right).
\end{align*}
Since $\vtheta_1$ is of length $nR$, sampling $\vtheta_1$ generally involves $\mO(n^3R^3)$ elementary operations. When both $n$ and $R$ are large, this sampling step could be computationally intensive. An alternative is to sample each $\vtheta_1^{(r)}$ separately, $r=1,\ldots, R$. For a fixed $r$, this can be done by substituting 
\[
	\vTheta_1\vTheta_{-1}' = \sum_{s=1}^R \vtheta_1^{(s)}\vtheta_{-1}^{(s)\prime},
\]
where $\vtheta_{-1}^{(s)} = \vtheta_3^{(s)}\otimes \vtheta_2^{(s)}$, into the transpose of \eqref{eq:var_A1} to obtain
\[
   \bY' = \sum_{s=1}^R\vtheta^{(s)}_{1}\vtheta_{-1}^{(s)\prime}\mX_{(1)}'+ \bU'. 
\]
Next, let $\bY_{1r} = \bY'-\sum_{s \neq r}\vtheta_{1}^{(s)}\vtheta_{-1}^{(s)\prime}\mX_{(1)}'$ and vectorize the above equation, we have
\[
	\tvec(\bY_{1r})= (\mX_{(1)}\vtheta^{(r)}_{-1} \otimes \bI_{n})\vtheta_{1}^{(r)}+\tvec(\bU').
\]
If the marginal prior for $\vtheta_1^{(r)}$ is
\[
	\vtheta_1^{(r)} \sim \distn{N}\left(\vtheta_{1,0}^{(r)}, \bV_{\vtheta_1^{(r)}}\right),
\]
the posterior distribution $(\vtheta_1^{(r)} \gvn \mathbf{Y}, \{\vtheta_{1}^{(s)}\}_{s\neq r}, \vtheta_{2}, \vtheta_{3}, \vSigma)$ has the form
\[
	(\vtheta_1^{(r)} \gvn \mathbf{Y}, \{\vtheta_{1}^{(s)}\}_{s\neq r}, \vtheta_{2}, \vtheta_{3}, \vSigma )
	\sim \distn{N}\left(\hat{\vtheta}_1^{(r)}, \bK_{\vtheta_1^{(r)}}^{-1}\right), 
\]
where
\begin{align*}
	\bK_{\vtheta_1^{(r)}} & = (\vtheta^{^{(r)}\prime}_{-1} \mX_{(1)}' \otimes \bI_{n})\vSigma^{-1}
	(\mX_{(1)}\vtheta^{(r)}_{-1} \otimes \bI_{n})\\
	\hat{\vtheta}_1^{(r)} & = \bK_{\vtheta_1^{(r)}}^{-1}\left(\bV_{\vtheta_1^{(r)}}^{-1} \vtheta_{1,0}^{(r)}
	+ (\vtheta^{(r)\prime}_{-1}\mX_{(1)}' \otimes \bI_{n})\vSigma^{-1}\tvec(\bY_{1r})\right).
\end{align*}
The drawback of sampling each $\vtheta_1^{(r)}$ separately is that this tends to increase the autocorrelation of the constructed Markov chain. But this approach is substantially faster and remains computationally feasible even when both $n$ and $R$ are large.

\subsection{Sampling of $\vtheta_2$ and $\vtheta_3$ }

Next, we derive the conditional distribution of $(\vtheta_2  \gvn \bY, \vtheta_{1}, \vtheta_3, \vSigma)$ and show that it is Gaussian. To start, we aim to write the VAR in \eqref{eq:var} as a linear regression in $\vtheta_2=\tvec(\vTheta_2')$. 

We first introduce some notations. Let $\be_i^p$ be the $i$-th column of $\mathbf{I}_p$. Let $\bP$ denote the $np \times np$  commutation matrix so that $\bP'\text{vec}(\bZ) = \text{vec}(\bZ')$ for any $n\times p$ matrix $\bZ$. Explicitly, $\bP$ can be constructed by setting the $(k,l)$ element to be 1, i.e., $\bP_{k,l} = 1$, if there exist $i$ and $j$ such that $k = (i-1)n+j$ and $l=(j-1)p + i$; otherwise, set $\bP_{k,l}=0$.

\begin{proposition}\rm \label{prop:1} The VAR in \eqref{eq:var} can be written as
\begin{equation} \label{eq:Theta2}
	\tvec(\bY') = \sum_{i=1}^{p}(\bP_{i2}\mX'_{(2)} \otimes \bP_{i1}\vTheta_{-2})\vtheta_2 + \tvec(\bU'),
\end{equation}
where $\bP_{i1}=(\bI_{n}\otimes (\be^p_i)')\bP'$, $\bP_{i2} = (\be^p_i)'\otimes\bI_{T}$  and $\mX_{(2)}$ is the mode-2 matricization of $\mX$.
\end{proposition}
The proof of this proposition is given in Appendix A. Now, given the representation in \eqref{eq:Theta2} and the Gaussian prior on $\vtheta_2$ specified in \eqref{eq:prior}, by standard linear regression results, one can verify that the full conditional posterior of $\vtheta_2$ is given by
\[
	(\vtheta_2  \gvn \bY, \vtheta_{1}, \vtheta_3, \vSigma) \sim 
		\mathcal{N}\left(\hat{\vtheta}_2, \mathbf{K}_{\vtheta_2}^{-1}\right), 
\]
where
\begin{align*}
	\mathbf{K}_{\vtheta_2} & = \mathbf{V}_{\vtheta_2}^{-1} + 
	 \sum_{i=1}^{p}\sum_{j=1}^{p}(\mX_{(2)}\bP'_{i2}\otimes \vTheta'_{-2}\bP'_{i1})\vSigma^{-1}(\bP_{j2}\mX'_{(2)}\otimes \bP_{j1}\vTheta_{-2}),\\ 
  \hat{\vtheta}_2 & = \mathbf{K}_{\vtheta_2}^{-1}\left(\mathbf{V}_{\vtheta_2}^{-1} \vtheta_{2,0} 
	+\sum_{i=1}^{p}(\mX_{(2)}\bP'_{i2}\otimes \vTheta'_{-2}\bP'_{i1})\vSigma^{-1}\tvec(\bY')\right).
\end{align*}
Therefore, one can sample $\vtheta_2$ in one block. When $n$ and $R$ are large, it might only be feasible to sample each $\vtheta_2^{(r)}$ at a time for $r=1,\ldots, R$. We provide the details of this alternative approach in Appendix C.

Likewise, the conditional distribution of $(\vtheta_3  \gvn \bY, \vtheta_{1}, \vtheta_2, \vSigma)$ can be shown to be Gaussian. More specifically, we first write the VAR in \eqref{eq:var} as a linear regression in $\vtheta_3=\tvec(\vTheta_3')$. To that end, let $\be_i^n$ denote the $i$-th column of $\bI_n$, let $\mX_{(3)}$ represent the mode-3 matricization of $\mX$, and define $\bQ$ to be the $n^2 \times n^2$ commutation matrix so that $\bQ'\text{vec}(\bZ) = \text{vec}(\bZ')$ for any $n\times n$ matrix $\bZ$. Explicitly, $\bQ$ can be constructed by setting the $\bQ_{k,l} = 1$ if there exist $i$ and $j$ such that $k = (i-1)n+j$ and $l = (j-1)n + i$; otherwise, set $\bQ_{k,l} = 0$. 

\begin{proposition}\rm \label{prop:2} The VAR in \eqref{eq:var} can be represented as
\begin{equation} \label{eq:Theta3}
	\tvec(\bY') = \sum_{i=1}^{n}(\bQ_{i2}\mX'_{(3)} \otimes \bQ_{i1}\vTheta_{-3})\vtheta_3 + \tvec(\bU'),
\end{equation}
where  $\bQ_{i1} = (\bI_{n}\otimes (\be^n_i)')\bQ'$ and $\bQ_{i2} = (\be^n_i)'\otimes\bI_{T}$.
\end{proposition}
Given the representation in \eqref{eq:Theta3}, it is easy to verify that 
\[
	(\vtheta_3  \gvn \bY, \vtheta_{1}, \vtheta_2, \vSigma) \sim 
	\mathcal{N}\left(\hat{\vtheta}_3, \bK_{\vtheta_3}^{-1}\right),
\]
where
\begin{align*}
	\bK_{\vtheta_3} & = \bV_{\vtheta_3}^{-1} + \sum_{i=1}^{n}\sum_{j=1}^{n}
	(\mX_{(3)}\bQ'_{i2}\otimes \vTheta'_{-3}\bQ'_{i1})\vSigma^{-1}(\bQ_{j2}\mX'_{(3)}\otimes \bQ_{j1}\vTheta_{-3})\\
	\hat{\vtheta}_3 & = \bK_{\vtheta_3}^{-1}\left(\bV_{\vtheta_3}^{-1} \vtheta_{3,0} 
	+ \sum_{i=1}^{n}(\mX_{(3)}\bQ'_{i2}\otimes \vTheta^{'}_{-3}\bQ'_{i1})\vSigma^{-1}\tvec(\bY')\right).
\end{align*}
Alternatively, one can sample each $\vtheta_3^{(r)}$ at a time for $r=1,\ldots, R$ when $n$ and $R$ are large as before.

Finally, sampling the time-varying error covariance matrices $\vSigma_1,\ldots, \vSigma_T$ naturally depends on the stochastic volatility specification used. For a wide variety of stochastic volatility specifications commonly employed in applied work, such as the common stochastic volatility and Cholesky stochastic volatility discussed earlier, efficient algorithms are available to sample the latent variables and model parameters. We refer the readers to \citet{chan23JE} for more details.

\section{An Empirical Application} \label{s:apps}

We conduct an out-of-sample forecasting exercise to evaluate the performance of the proposed tensor VARs compared to a standard benchmark. More specifically, we construct a dataset of 40 quarterly macroeconomic and financial variables from the FRED-QD database \citep{MN20}. The sample spans from 1969Q1 to 2024Q1. The data include key macroeconomic variables such as GDP, inflation rates, labor market variables and various interest rates. We refer the readers to Appendix B for the detailed description of the time series and their transformations. 

We consider three tensor VARs: a homoskedastic tensor VAR (TVAR), TVARs with the common volatility (TVAR-CSV) and the Cholesky stochastic volatility (TVAR-SV). For now, we set the rank of the VAR coefficient tensor $\mA$ to be $R=1$ for all TVARs. As a benchmark, we also include a standard Bayesian VAR (BVAR) with the Minnesota prior (implemented as the natural conjugate prior). The evaluation period of the forecasting exercise begins in 2010Q1 and ends in 2024Q1.

To assess the performance of jointly forecasting all $n=40$ time series, we calculate the average log predictive likelihoods for each model over one- and four-quarter-ahead forecast horizons. The results are presented in Table~\ref{tab:dens1}. Higher values of log predictive likelihoods signify better forecast performance.
\begin{table}[H]
\centering
\caption{Joint density forecast performance of the proposed tensor VARs relative to a standard Bayesian VAR with the Minnesota prior.} \label{tab:dens1} 
\begin{tabular}{lcccc}
\hline\hline               
           & BVAR  & TVAR & TVAR-CSV & TVAR-SV            \\ \hline
One-quarter-ahead  & $-$8.01 & $-$3.30   & $-$1.31  & $-$2.32  \\
\rowcolor{lightgray}
Four-quarter-ahead & $-$8.53 & $-$3.30   & $-$2.21  & $-$2.66  \\
\hline\hline
\end{tabular}
\end{table}

A few observations can be drawn from these forecasting results. Firstly, comparing the two homoskedastic models, TVAR and BVAR, it is clear that the former substantially outperforms the latter, suggesting that the specification of a low rank structure on the VAR coefficient tensor $\mA$ is more appropriate. Secondly, allowing some form of time-varying volatility clearly improves forecast performance. For example, extending a homoskedastic TVAR to a version with the common stochastic volatility increases the average log predictive likelihood from $-8$ to about $-3$ for one-step-ahead density forecasts. This finding is line with the large body of empirical evidence that demonstrates the importance of allowing time-varying volatility in macroeconomic forecasting \citep{clark11, DGG13, chan23JE}. Finally, among the two TVARs with stochastic volatility, the version with the parsimonious common stochastic volatility outperforms the one with the more flexible Cholesky stochastic volatility, indicating strong comovements in the macroeconomic volatilities. 

Next, we look at the point forecast performance of the TVARs for individual time series. In particular, Table \ref{tab:point1} reports the root mean squared forecast errors (RMSFEs) of a few key macroeconomic variables over the evaluation period relative to the benchmark BVAR. Values less than one indicate better forecast performance than the benchmark.

\begin{table}[H]
\caption{Root mean squared forecast errors of the proposed tensor VARs relative to a standard Bayesian VAR with the Minnesota prior.}
 \label{tab:point1} 
\centering
\resizebox{\textwidth}{!}{
\begin{tabular}{lcccccc}
\hline\hline
  & \multicolumn{3}{c}{One-quarter-ahead} & \multicolumn{3}{c}{Four-quarter-ahead} \\ 
  & TVAR    & TVAR-CSV    & TVAR-SV  & TVAR  & TVAR-CSV  & TVAR-SV      \\ \hline
RPI                     & 0.93         & 0.93         & 0.93         & 0.99        & 0.99          & 0.98  \\
\rowcolor{lightgray}
GDP                     & 0.65         & 0.64         & 0.78         & 0.82        & 0.82          & 0.85  \\
Unemployment            & 0.68         & 0.68         & 0.72         & 0.86        & 0.86          & 0.86  \\
\rowcolor{lightgray}
CPI                     & 1.03         & 1.02         & 1.14         & 0.89        & 0.89          & 0.93  \\
Fed funds rate          & 0.83         & 0.83         & 1.46         & 0.60        & 0.60          & 0.65  \\
\rowcolor{lightgray}
10-year T-bond          & 0.92         & 0.92         & 0.95         & 0.82        & 0.82          & 0.85    \\ \hline\hline            
\end{tabular}}
\end{table}

The results show that the TVARs tend to outperform the benchmark BVAR for both one- and four-quarter-ahead horizons. For example, the RMSFE of the homoskedastic TVAR for forecasting GDP is only 65\% of that of the benchmark. Interestingly, allowing time-varying volatility does not appear to improve point forecasts. In fact, the more flexible TVAR-SV often provides inferior point forecasts compared to the homoskedastic TVAR (though the performance of the more parsimonious TVAR-CSV is virtually identical to TVAR).

So far we have set $R$, the rank of the VAR coefficient tensor $\mA$, to be 1 for all TVARs. Naturally, one might wonder how this choice affects the density and point forecast performance. To investigate the impact of the choice of $R$, we present in Table~\ref{tab:dens2} the one-quarter-ahead average log predictive likelihoods of the three TVARs for $R=1,3,5,10$.

\begin{table}[H]
\caption{One-quarter-ahead joint density forecast performance of the proposed tensor VARs for $R=1,3,5,10$.} 
\label{tab:dens2} 
\centering
\begin{tabular}{lccc}\hline\hline
       & TVAR   & TVAR-CSV & TVAR-SV \\ \hline\hline
$R=1$  & $-$3.30  & $-$1.31    & $-$2.32   \\
\rowcolor{lightgray}
$R=3$  & $-$3.27  & $-$1.39    & $-$2.36   \\
$R=5$  & $-$3.40  & $-$1.26    & $-$2.48   \\
\rowcolor{lightgray}
$R=10$ & $-$3.53  & $-$1.29    & $-$2.53   \\ \hline\hline
\end{tabular}
\end{table}

It is interesting to note that the choice of $R$ does not seem to have a large impact on the joint density forecast performance across the three TVARs. In particular, it does not affect the relative ranking of the three models; the two TVARs with time-varying volatility outperform the homoskedastic version regardless of the rank $R$. Overall, low-rank TVARs tend to work as well as, if not better than, TVARs with $R=10$.

Table~\ref{tab:point2} reports the one-quarter-ahead RMSFEs of the three TVARs for forecasting a few key macroeconomic variables. To facilitate comparison, the results for each TVAR are relative to the corresponding TVAR with $R=1$. With the exception of CPI inflation, increasing the rank $R$ does not appear to substantially improve point forecast performance. 

\begin{table}[H]
\caption{One-quarter-ahead RMSFEs of TVARs with $R=3,5,10$ relative to the corresponding TVARs with $R=1$.}
 \label{tab:point2} 
\centering
\resizebox{\textwidth}{!}{
\begin{tabular}{lccccccccc}
\hline\hline
    & \multicolumn{3}{c}{$R=3$} & \multicolumn{3}{c}{$R=5$}  & \multicolumn{3}{c}{$R=10$} \\ 
    & TVAR   & TVAR-CSV    & TVAR-SV   & TVAR   & TVAR-CSV    & TVAR-SV & TVAR   & TVAR-CSV    & TVAR-SV \\ \hline
RPI           & 1.01   & 1.00  & 1.01   & 1.01  & 1.00  & 1.02   & 1.01   & 1.00  & 1.03 \\
\rowcolor{lightgray}
GDP           & 1.04   & 1.02  & 0.99   & 1.08  & 1.05  & 1.05   & 1.12   & 1.07  & 1.12  \\
Unemployment  & 1.02   & 1.01  & 1.00   & 1.02  & 1.01  & 1.05   & 1.05   & 1.01  & 1.15  \\
\rowcolor{lightgray}
CPI inflation & 0.97   & 0.95  & 1.06   & 0.96  & 0.95  & 0.90   & 0.89   & 0.90  & 0.94  \\
Fed funds rate& 1.08   & 1.11  & 1.08   & 1.24  & 1.21  & 1.05   & 1.20   & 1.18  & 0.99  \\
\rowcolor{lightgray}
10-year T-bond& 1.00   & 1.00  & 1.22   & 0.99  & 1.00  & 1.17   & 0.98   & 0.99  & 1.16  \\ \hline\hline  
\end{tabular}}
\end{table}

Overall, these forecasting results demonstrate the benefits of the proposed approach of specifying a low rank structure on the VAR coefficient tensor. In addition, our results highlight the importance of accommodating time-varying volatility in forecasting macroeconomic variables. Notably, in our forecasting exercise, the more parsimonious TVAR-CSV tends to forecast better than the more flexible TVAR-SV.

\section{Concluding Remarks and Future Research} \label{s:conclusions}

We have developed Bayesian tensor VARs in which the VAR coefficients are arranged as a third-order tensor and parameterized using a rank-$R$ CP decomposition. We then introduced efficient sampling algorithms to simulate the components of the tensor decomposition. Through a forecasting exercise, we showed that these TVARs outperform the standard BVAR with the Minnesota prior. 

For future research, it would be useful to extend these TVARs to allow for time-varying coefficients in the mean equations. This tensor framework is especially suitable for developing time variation in the VAR coefficients, as the number of free parameters grows only linear in $n$.
Of course, for large systems one might need additional shrinkage. In those cases, one can consider either the static shrinkage approach in \citet{chan23} or the dynamic shrinkage approaches proposed in \citet{KK18} and \citet{HKO19}.

\newpage

\section*{Appendix A: Proofs of Propositions}

In this appendix we provide a proof of Proposition~\ref{prop:1}. The proof of Proposition~\ref{prop:2} is very similar and is thus omitted. We first introduce some useful notations. Let $\be_i^m$ denote the $i$-th column of $\bI_m$; if there is no ambiguity about the dimension, we simply write $\be_i$. For a generic third-order tensor $\mZ$, let $\mZ_{(k)}$ denote its mode-$k$ matricization for $k=1,2,3$. We use the notation $\mZ_{i_1,i_2,i_3}$ to represent the $(i_1,i_2,i_3)$ tensor element of $\mZ$ and 
$\mZ_{(k),i,j}$ to denote the $(i,j)$ element of $\mZ_{(k)}$.

As an example, consider the third-order tensor $\mX\in\mathbb{R}^{T\times n\times p}$. Its mode-2 matricization $\mX_{(2)}$ is an $n\times Tp$ matrix, and the $(i_1,i_2,i_3)$ tensor element 
$\mX_{i_1,i_2,i_3}$  maps to the $(i_2,j)$ matrix element $\mX_{(2),i_2,j}$, where $j = (i_3-1)T + i_1$. Conversely, the $(i,j)$ matrix element maps to  the $(i_1,i,i_3)$ tensor element, where $i_1 = \text{mod}(j,T)$ is the remainder after division of $j$ by $T$, and $i_3 = \lceil{j/T}\rceil$ is the least integer greater than or equal to $j/T$. It can be easily verified that $j = (i_3-1)T + i_1$.

Naturally, any matrix multiplication involving matrices constructed by mode-$k$ matricization can be written in terms of the elements in the original tensors. For example, consider the $(i,j)$ element of the matrix $\bE = \mX'_{(2)} \mA_{(2)}$. Let $i_3 = \lceil{i/T}\rceil$ and $i_1 = \text{mod}(i,T)$ so that $i = (i_3-1)T + i_1$. Similarly, obtain integers $j_1$ and $j_3$ so that $j = (j_3-1)n + j_1$. Then, we can write $\bE_{i,j}$ in terms of elements in $\mX$ and $\mA$:
\begin{equation} \label{eq:E}
\begin{split}
	\bE_{i,j} & = \sum_{k = 1}^{n} \mX_{(2),k,i} \mA_{(2),k,j} \\
	& =\sum_{k = 1}^{n} \mX_{(2),k,(i_3-1)T + i_1} \mA_{(2),k,(j_3-1)T + j_1}
	= \sum_{k = 1}^{n} \mX_{i_1,k,i_3} \mA_{j_1,k,j_3}.
\end{split}
\end{equation}

Next, we introduce a useful lemma. 
\begin{lemma}\label{lemma:3} \rm The VAR in \eqref{eq:var} can be written as
\begin{equation} \label{eq:tildeU}
	\bY = \sum_{i=1}^{p}(\be'_{i}\otimes \bI_T)\bE\bP(\bI_n\otimes \be_{i}) + \bU,
\end{equation}
where $\be_i \equiv \be_i^p$ is the $i$-th column of $\bI_p$, $\bE = \mX'_{(2)} \mA_{(2)}$, and $\bP$ is the $np \times np$ commutation matrix in which $\bP_{k,l} = 1$ if there exist integers $i$ and $j$ such that $k = (i-1)n+j$ and $l=(j-1)p +i$; otherwise, $\bP_{k,l}=0$.
\end{lemma}

\begin{proof}[Proof of Lemma \ref{lemma:3}]  First, note that from 
\[
	\bY = \mX_{(1)}\mA_{(1)}'+\bU,
\]
the element $\bY_{t,l}$ can be expressed as
\[
	\bY_{t,l} = \sum_{i=1}^{np}\mX_{(1),t,i}\mA_{(1),l,i} + \bU_{t,l} 
	= \sum_{i_1=1}^{n}\sum_{i_3=1}^{p}\mX_{t,i_1,i_3}\mA_{l,i_1,i_3}+\bU_{t,l},
\]
where $i = (i_3-1)n+i_1$. Now, it suffices to show that the $(t,l)$ element of the sum on the right-hand side of equation \eqref{eq:tildeU} has the same expression and is thus equal to $\bY_{t,l}$. 

We start by describing the typical element of the $Tp \times np$ matrix $\bE\bP$. For any $(i,j)$, obtain the pairs of integers $(i_1,i_3)$ and $(j_1,j_3)$ such that $i= (i_3-1)T+i_1$ and $j=(j_3-1)p+j_1$. Then, by direct computation, we have
\begin{align*}
    (\bE\bP)_{i,j} = &(\bE\bP)_{(i_3-1)T+i_1,(j_3-1)p+j_1}\\
		 = & \sum_{s=1}^{np} \bE_{(i_3-1)T+i_1,s}\bP_{s,(j_3-1)p+j_1}\\
		 = & \sum_{k=1}^{p}\sum_{l=1 }^n \bE_{(i_3-1)T+i_1,(k-1)n+l} \bP_{(k-1)n+l,(j_3-1)p+j_1}\\
		  = &\sum_{k=1}^{p}\sum_{l=1 }^n \bE_{(i_3-1)T+i_1,(k-1)n + l}1(k=j_1,l=j_3)\\
			= &\bE_{(i_3-1)T+i_1,(j_1-1)n+j_3}.
\end{align*}
Then, the $(t,l)$ element of the sum on the right-hand side  of equation \eqref{eq:tildeU} is given by:
\begin{align*}
    & \sum_{k=1}^{p} \sum_{i=1}^{Tp}\sum_{j = 1}^{np}(\be'_{k}\otimes \bI_T)_{t,i}(\bE\bP)_{i,j} 
		(\bI_n\otimes \be_{k})_{j,l}+\bU_{t,l}\\		
		& = \sum_{k=1}^{p}\sum_{i_1=1}^{T}\sum_{i_3=1}^{p}\sum_{j_3=1}^{n}\sum_{j_1=1}^{p}
		(\be'_{k}\otimes \bI_T)_{t,(i_3-1)T+i_1}\bE_{(i_3-1)T+i_1,(j_1-1)n+j_3}
		(\bI_n\otimes \be_{k})_{(j_3-1)p+j_1,l}+\bU_{t,l}\\
		& = \sum_{k=1}^{p}\sum_{i_1=1}^{T}\sum_{i_3=1}^{p}\sum_{j_3=1}^{n}\sum_{j_1=1}^{p}(\be'_{{k},i_3} \bI_{T,t,i_1})\bE_{(i_3-1)T+i_1,(j_1-1)n+j_3}(\bI_{n,j_3,l} \be_{{k},j_1})+\bU_{t,l}\\
		& = \sum_{k=1}^{p}\sum_{i_1=1}^{T}\sum_{i_3=1}^{p}\sum_{j_3=1}^{n}\sum_{j_1=1}^{p}1({k}=i_3,t=i_1)\bE_{(i_3-1)T+i_1,(j_1-1)n+j_3}1(j_3=l,k=j_1)+\bU_{t,l}\\
		& = \sum_{k=1}^{p}\bE_{(k-1)T+t,(k-1)n+l}+\bU_{t,l}\\
		& = \sum_{k=1}^{p}\sum_{s=1}^{n} \mX_{t,s,k} \mA_{l,s,k} + \bU_{t,l},
\end{align*}
where the last equality holds because of \eqref{eq:E}. Hence, we have shown that the $(t,l)$ element of the sum on the right-hand side is $\bY_{t,l}$, thus completing the proof.
\end{proof}

\begin{proof}[Proof of Proposition \ref{prop:1}] 

We begin by taking the transpose of both sides of equation $\eqref{eq:tildeU}$ in Lemma \ref{lemma:3}:
\[
	\bY'=\sum_{i=1}^{p}(\bI_n \otimes \be'_{i})\bP'\bE'(\be_{i}\otimes \bI_T)+\bU'.
\]
Next, vectorizing both sides of the above equation, we have
\begin{align}
    \tvec(\bY')&=\sum_{i=1}^{p}\left((\be'_{i}\otimes \bI_T)\otimes (\bI_n \otimes \be'_{i})\bP'\right)
		\tvec(\bE')+ \tvec(\bU') \nonumber \\		
		& = \sum_{i=1}^{p}(\bP_{i2}\otimes \bP_{i1})\tvec(\bE')+ \tvec(\bU') \nonumber \\
		& = \sum_{i=1}^{p}(\bP_{i2}\otimes \bP_{i1})\tvec(\mA_{(2)}'\mX_{(2)} )+ \tvec(\bU') \nonumber \\
		& = \sum_{i=1}^{p}(\bP_{i2}\otimes \bP_{i1})\tvec( \vTheta_{-2}\vTheta'_{2}\mX_{(2)})
    + \tvec(\bU') \label{eq:proof1} \\
		& = \sum_{i=1}^{p}(\bP_{i2}\otimes \bP_{i1})(\mX'_{(2)}\otimes \vTheta_{-2} )\tvec(\vTheta_2') 
		+ \tvec(\bU')\nonumber \\
		& = \sum_{i=1}^{p}(\bP_{i2}\mX'_{(2)}\otimes \bP_{i1}\vTheta_{-2})\vtheta_2+ \tvec(\bU'),\nonumber	
\end{align}
where $\bP_{i1}=(\bI_{n}\otimes \be'_i)\bP'$ and $\bP_{i2}=\be'_i\otimes\bI_{T}$, $\vtheta_2=\tvec(\vTheta_2')$.
The fourth equality holds because $\mA_{(2)} = \vTheta_2\vTheta_{-2}'.$
\end{proof}

\newpage

\section*{Appendix B: Data}

The quarterly dataset is sourced from the FRED-QD database \citep{MN20} maintained by the Federal Reserve Bank of St. Louis. The sample spans from 1969Q1 to 2024Q1. Table~\ref{tab:data} lists the variables and describes how they are transformed. To ensure consistency, each time series is standardized to have zero mean and unit variance.

\begin{table}[H]
\centering
\caption{A list of variables and their transformations. 
Tcode: 1: no transformation; 2: $\Delta x_t$; 5: $\Delta \log(x_t)$; 6: $\Delta^2 \log(x_t)$.}
\label{tab:data}
{\small\begin{tabular}{llc}
\hline\hline
Name & Description & T-code \\ \hline
RPI                                 & Real Personal Income                      & 5 \\
\rowcolor{lightgray}
INDPRO                              & IP Index                                  & 5 \\
GDP                                 & Real Gross Domestic Product               & 5 \\
\rowcolor{lightgray}
GDPDEFL                             & GDP deflator                              & 6 \\
DPCERA3M086SBEA                     & Real PCE                                  & 5 \\
\rowcolor{lightgray}
CMRMTSPLx                           & Real M \& T Sales                         & 5 \\
HWI                                 & Help-Wanted Index for US                  & 2 \\
\rowcolor{lightgray}
HWIURATIO                           & Help Wanted to Unemployed ratio           & 2 \\
CLF16OV                             & Civilian Labor Force                      & 5 \\
\rowcolor{lightgray}
UNRATE                              & Civilian Unemployment Rate                & 2 \\
PAYEMS                              & All Employees: Total nonfarm              & 5 \\
\rowcolor{lightgray}
CES0600000007                       & Hours: Goods-Producing                    & 5 \\
CPIAUCSL                            & CPI: All Items                            & 6 \\
\rowcolor{lightgray}
FEDFUNDS                            & Effective Federal Funds Rate              & 2 \\
TB3MS                               & 3-Month T-bill                            & 2 \\
\rowcolor{lightgray}
TB6MS                               & 6-Month T-bill                            & 2 \\
GS1                                 & 1-Year T-bond                             & 2 \\
\rowcolor{lightgray}
GS5                                 & 5-Year T-bond                             & 2 \\
GS10                                & 10-Year T-bond                            & 2 \\
\rowcolor{lightgray}
AAA                                 & Aaa Corporate Bond Yield                  & 2 \\
BAA                                 & Baa Corporate Bond Yield                  & 2 \\
\rowcolor{lightgray}
M1SL                                & M1 Money Stock                            & 5 \\
M2SL                                & M2 Money Stock                            & 5 \\
\rowcolor{lightgray}
BUSLOANS                            & Commercial and Industrial Loans           & 5 \\
NONREVSL                            & Total Nonrevolving Credit                 & 5 \\
\rowcolor{lightgray}
INVEST                              & Securities in Bank Credit                 & 5 \\
S\&P   500                          & S\&P 500                                  & 5 \\
\rowcolor{lightgray}
S\&P   div yield                    & S\&P Dividend yield                       & 2 \\
S\&P PE   ratio                     & S\&P Price/Earnings ratio                 & 5 \\
\rowcolor{lightgray}
EXSZUSx                             & Switzerland / U.S. FX Rate                & 5 \\
EXJPUSx                             & Japan / U.S. FX Rate                      & 5 \\
\rowcolor{lightgray}
EXUSUKx                             & U.S. / U.K. FX Rate                       & 5 \\
EXCAUSx                             & Canada / U.S. FX Rate                     & 5 \\
\rowcolor{lightgray}
UEMPMEAN                            & Average Duration of   Unemployment        & 5 \\
AWHMAN                              & Hours: Manufacturing                      & 1 \\
\rowcolor{lightgray}
ISRATIOx                            & Inventories to Sales Rati                 & 2 \\
REALLN                              & Real Estate Loans                         & 5 \\
\rowcolor{lightgray}
PPICMM                              & PPI: Commodities                          & 6 \\
PCEPI                               & PCE: Chain-type Price   Index             & 6 \\
\rowcolor{lightgray}
FPI                                 & Fixed Private Investment                  & 5 \\  \hline\hline
\end{tabular}}
\end{table}

\newpage

\section*{Appendix C: Additional Estimation Details}

In this appendix we provide technical details on the alternative approach of sampling each $\vtheta_2^{(r)}$ at a time for $r=1,\ldots, R$. To that end, we derive the conditional distribution of $\vtheta_2^{(r)}$ given $\{\vtheta_2^{(s)}\}_{s\neq r}$ and other model parameters.

We assume the following the marginal prior for $\vtheta_2^{(r)}$:
\[
	\vtheta_2^{(r)} \sim \distn{N}\left(\vtheta_{2,0}^{(r)}, \bV_{\vtheta_2^{(r)}}\right).
\]
Next, it follows from equation \eqref{eq:proof1} in the proof of Proposition~1 that
\begin{align*}
\tvec(\bY') & = \sum_{i=1}^{p}(\bP_{i2}\otimes \bP_{i1})\tvec( \vTheta_{-2}\vTheta'_{2}\mX_{(2)})
    + \tvec(\bU') \\
					& =\sum_{i=1}^{p}\tvec(\bP_{i1}\vTheta_{-2}\vTheta'_{2}\mX_{(2)}\bP_{i2}')+ \tvec(\bU').
\end{align*}  
Noting that $\vTheta_2\vTheta_{-2}' = \sum_{s=1}^R\vtheta_2^{(s)}\vtheta_{-2}^{(s)\prime}$, where $\vtheta_{-2}^{(s)} = \vtheta_{3}^{(s)}\otimes \vtheta_{1}^{(s)},$ one can express $\bY'$ as:
\begin{align*}
		 \bY' & = \sum_{i=1}^{p}\bP_{i1}\vTheta_{-2}\vTheta'_{2}\mX_{(2)}\bP_{i2}'+ \bU' \\
			    & = \sum_{i=1}^{p}\sum_{s=1}^R\bP_{i1}\vtheta_{-2}^{(s)}\vtheta_{2}^{(s)\prime}
					\mX_{(2)}\bP_{i2}' + \bU'.
\end{align*}

Fix $r$ and let $\bY_{2r} = \bY' - \sum_{i=1}^{p}\sum_{s\neq r}\bP_{i1}\vtheta_{-2}^{(s)}
\vtheta_{2}^{(s)\prime}\mX_{(2)}\bP_{i2}'$. Vectorize the above equation to get
\[
	\tvec(\bY_{2r}) = \sum_{i=1}^{p}(\bP_{i2}\mX_{(2)}'\otimes\bP_{i1}\vtheta_{-2}^{(r)})
	\vtheta_{2}^{(r)}+\tvec(\bU').
\]

Therefore, the posterior conditional distribution $(\vtheta_2^{(r)} \gvn \mathbf{Y}, \vtheta_1, \{\vtheta_{2}^{(s)}\}_{s\neq r}, \vtheta_{3}, \vSigma)$ can be expressed as
\[
	(\vtheta_2^{(r)} \gvn \mathbf{Y}, \vtheta_1, \{\vtheta_{2}^{(s)}\}_{s\neq r}, 
	\vtheta_{3}, \vSigma)\sim\distn{N}\left(\hat{\vtheta}_2^{(r)}, \bK_{\vtheta_2^{(r)}}^{-1}\right), 
\]
where
\begin{align*}
	\bK_{\vtheta_2^{(r)}} & = \bV_{\vtheta_2^{(r)}}^{-1} 	+ \sum_{i=1}^{p}\sum_{j=1}^{p}
	(\mX_{(2)}\bP'_{i2}\otimes\vtheta^{(r)\prime}_{-2}\bP'_{i1})\vSigma^{-1}
	(\bP_{j2}\mX'_{(2)}\otimes \bP_{j1}\vtheta^{(r)}_{-2}) \\
	\hat{\vtheta}_2^{(r)} & = \bK_{\vtheta_2^{(r)}}^{-1}\left(\bV_{\vtheta_2^{(r)}}^{-1}\vtheta_{2,0}^{(r)} 
	+ \sum_{i=1}^{p}(\mX_{(2)}\bP'_{i2}\otimes \vtheta^{(r)\prime}_{-2}\bP'_{i1})
	\vSigma^{-1}\tvec(\bY_{2r})\right).
\end{align*}

\newpage

\singlespace

\ifx\undefined\BySame
\newcommand{\BySame}{\leavevmode\rule[.5ex]{3em}{.5pt}\ }
\fi
\ifx\undefined\textsc
\newcommand{\textsc}[1]{{\sc #1}}
\newcommand{\emph}[1]{{\em #1\/}}
\let\tmpsmall\small
\renewcommand{\small}{\tmpsmall\sc}
\fi

\end{document}